\documentclass[conference]{IEEEtran}
\ifCLASSINFOpdf
\else
\fi
\usepackage{amsmath,amssymb,amsthm}

\usepackage{authblk}
\newtheorem{theorem}{Theorem}
\usepackage{hhline}
\usepackage{graphicx}
\usepackage{epstopdf}
\usepackage{subcaption}
\usepackage{algorithm}
\usepackage{algorithmic}
\usepackage{amsmath}
\usepackage{cite}

\begin{document}
	\title{Providing Wireless Coverage to High-rise Buildings Using UAVs}
	\author[1]{Hazim Shakhatreh}
	\author[1]{Abdallah Khreishah}
	\author[2]{Bo Ji}
	\affil[1]{Department of Electrical and Computer Engineering, New Jersey Institute of Technology}
	\affil[2]{Department of Computer and Information Sciences, Temple University }
	\maketitle
	
	\begin{abstract}
		Unmanned aerial vehicles (UAVs) can be used as aerial wireless base stations when cellular networks go down. Prior studies on UAV-based wireless coverage typically consider an Air-to-Ground path loss model, which assumes that the users are outdoor and they are located on a 2D plane. In this paper, we propose using a single UAV to provide wireless coverage for indoor users inside a high-rise building under disaster situations (such as earthquakes or floods), when cellular networks are down. First, we present a realistic Outdoor-Indoor path loss model and describe the tradeoff introduced by this model. Then, we study the problem of efficient UAV placement, where the objective is to minimize the total transmit power required to cover the entire high-rise building. The formulated problem is non-convex and is generally difficult to solve. To that end, we consider two cases of practical interest and provide the efficient solutions to the formulated problem under these cases. In the first case, we aim to find the minimum transmit power such that an indoor user with the maximum path loss can be covered. In the second case, we assume that the locations of indoor users are symmetric across the dimensions of each floor. 
	\end{abstract}
	
	\begin{IEEEkeywords}
		Unmanned aerial vehicles, Outdoor-to-Indoor path loss model.
	\end{IEEEkeywords}
	
	\section{Introduction}
	\label{sec:Introduction}
	UAVs can be used to provide wireless coverage during emergency cases where each UAV serves as an aerial wireless base station when the cellular network goes down~\cite{bupe2015relief}. They can also be used to supplement the ground base station in order to provide better coverage and higher data rates for the users~\cite{bor2016efficient}.
	
	In order to use a UAV as an aerial wireless base station, the authors in~\cite{al2014modeling} presented an Air-to-Ground path loss model that helped the academic researchers to formulate many important problems. The authors of~\cite{mozaffari2015drone,mozaffari2016optimal,mozaffari2016efficient,kalantari2016number,shakhatreh2016continuous} utilized this model to study the problem of UAV placement, where the objective is to minimize the number of UAVs for covering a given area.
	The authors of~\cite{mozaffari2015drone} described the tradeoff in this model. At a low altitude, the path loss between the UAV and the ground user decreases, while the probability of line of sight links also decreases. On the other hand, at a high altitude line of sight connections exist with a high probability, while the path loss increases.
	However, it is assumed that all users are outdoor and the location of each user can be represented by an outdoor 2D point. These assumptions limit the applicability of this model when one needs to consider indoor users.
	
	Providing good wireless coverage for indoor users is very important. According to Ericsson report~\cite{ericsson}, 90\% of the time people are indoor and 80\% of the mobile Internet access traffic also happens indoors~\cite{alcatel,cisco}. To guarantee the wireless coverage, the service providers are faced with several key challenges, including providing service to a large number of indoor users and the ping pong effect due to interference from near-by macro cells~\cite{comm,amplitic,zhang2016study}. In this paper, we propose using a single UAV to provide wireless coverage for users inside a high-rise building during emergency cases, when the cellular network service is not available.

	To the best of our knowledge, this is the first work that proposes using a UAV to provide wireless coverage for indoor users. We summarize our main contributions as follows. First, we assume an Outdoor-Indoor path loss model~\cite{series2009guidelines}, certified by ITU, and show the tradeoff introduced by this model. Second, we formulate the problem of efficient UAV placement, where the objective is to minimize the total transmit power required to cover the entire high-rise building. Third, since the formulated problem is non-convex and is generally difficult to solve, we consider two cases of practical interest and provide the efficient solutions to the formulated problem under these cases. In the first case, we aim to find the minimum transmit power such that an indoor user with the maximum path loss can be covered. In the second case, we assume that the locations of indoor users are symmetric across the dimensions of each floor, and propose a gradient descent algorithm for finding the efficient location of the UAV.

	The rest of this paper is organized as follows. In Section II, we describe the system model and a path loss model suitable for studying indoor wireless coverage. In Section III, we formulate the problem of UAV placement with an objective of minimizing the transmit power for covering the entire building. In Section IV, we describe the tradeoff introduced by the path loss model and show how to find the efficient location of the UAV such that the total transmit power is minimized in two scenarios of practical interest. Finally, we present our numerical results in Section V and make concluding remarks in Section VI.
\begin{figure*}[ht]
	\begin{minipage}[b]{0.25\linewidth}
		\centering
		\includegraphics[width=\textwidth]{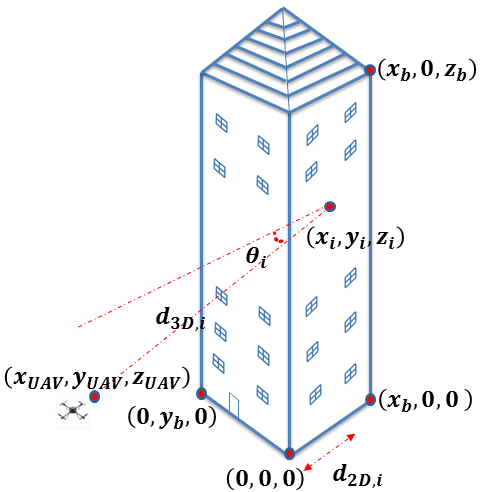}
		\caption{Parameters of the            
			path loss model
			}
		\label{fig:figure1}
	\end{minipage}
	\hspace{0.1cm}
	\begin{minipage}[b]{0.355\linewidth}
		\centering
		\includegraphics[width=\textwidth]{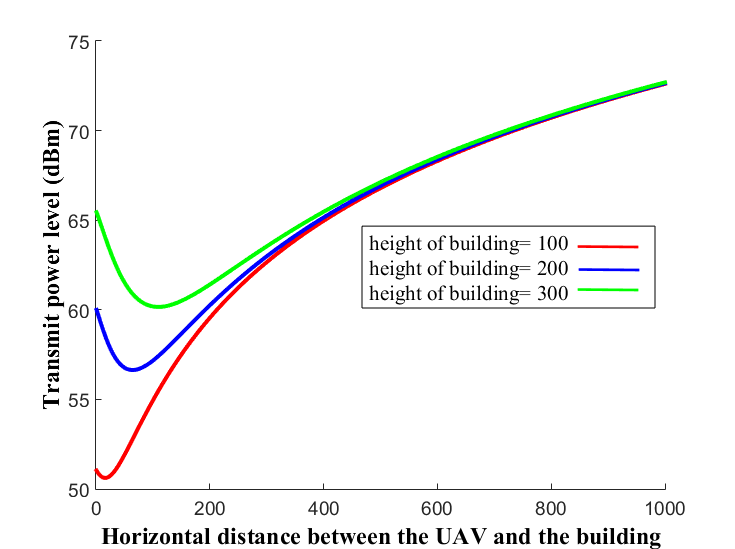}
		\caption{Transmit power required to cover the building
			}
		\label{fig:figure2}
	\end{minipage}
	\hspace{0.1cm}
	\begin{minipage}[b]{0.355\linewidth}
		\centering
		\includegraphics[width=\textwidth]{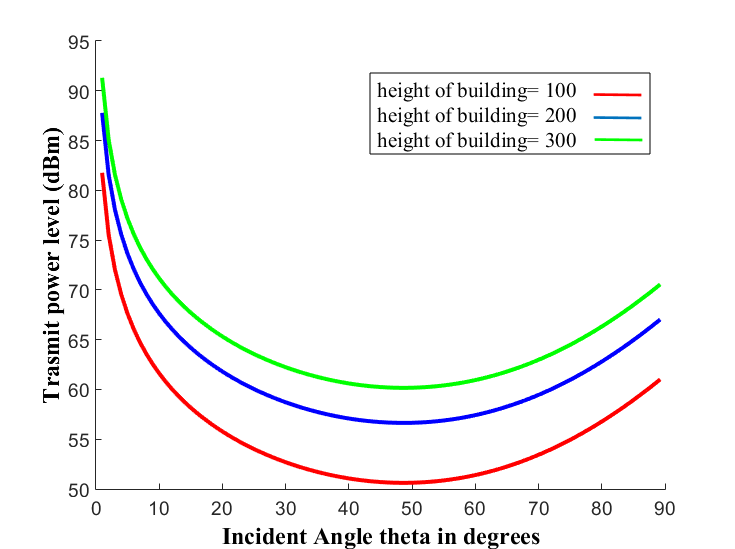}
		\caption{Transmit power required to cover the building 
			}
		\label{fig:figure3}
	\end{minipage}
\end{figure*}
	
	\section{System Model}
	\label{sec:system_model}
	\subsection{System Settings}
	\label{subsec:system settings}
	Let ($x_{UAV}$,$y_{UAV}$,$z_{UAV}$) denote the 3D location of the UAV. We assume that all users are located inside a high-rise building as shown in Figure~\ref{fig:figure1}, and use ($x_{i}$,$y_{i}$,$z_{i}$) to denote the location of user $i$. The dimensions of the high-rise building are $[0,x_b]$ $\times$ $[0,y_b]$ $\times$ $[0,z_b]$. Also, let $d_{3D,i}$ be the 3D distance between the UAV and indoor user $i$, let $\theta_{i}$ be the incident angle , and let $d_{2D,i}$ be the 2D indoor distance of user $i$ inside the building.
	
	\subsection{Outdoor-Indoor Path Loss Model}
	\label{subsec:Outdoor-Indoor Path Loss Model}
	The Air-to-Ground path loss model presented in~\cite{al2014modeling} is not appropriate when we consider wireless coverage for indoor users, because this model assumes that all users are outdoor and located at 2D points. In this paper, we adopt the Outdoor-Indoor path loss model, certified by the ITU~\cite{series2009guidelines}. The path loss is given as follows:\\
	\begin{equation}
	\begin{split}
	L_i=L_{F}+L_{B}+L_{I}= ~~~~~~~~~~~~~~~~\\
	(w\log_{10}d_{3D,i}+w\log_{10}f_{Ghz}+g_{1})+\\
	(g_{2}+g_{3}(1-\cos\theta_{i})^{2})+(g_{4}d_{2D,i})~~~~
	\end{split}
	\end{equation}
	where $L_{F}$ is the free space path loss, $L_{B}$ is the building penetration loss, and $L_{I}$ is the indoor loss. In this model, we also have
	$w$=20, $g_{1}$=32.4, $g_{2}$=14, $g_{3}$=15, $g_{4}$=0.5~\cite{series2009guidelines} and $f_{Ghz}$ is the carrier frequency (2Ghz). Note that there is a key tradeoff in the above model when the horizontal distance between the UAV and a user changes. When this horizontal distance increases, the free space path loss (i.e., $L_F$) increases as $d_{3D,i}$ increases, while the building penetration loss (i.e., $L_B$) decreases as the incident angle (i.e., $\theta_i$) decreases. Similarly, when this horizontal distance decreases, the free space path loss (i.e., $L_F$) decreases as $d_{3D,i}$ decreases, while the building penetration loss (i.e., $L_B$) increases as the incident angle (i.e., $\theta_i$) increases. 
	\section{Problem Formulation }
	Consider a transmission between a UAV located at ($x_{UAV}$,$y_{UAV}$,$z_{UAV}$) and an indoor user $i$ located at ($x_i$,$y_i$,$z_i$). The rate for user $i$ is given by:
	\begin{equation}
	\begin{split}
			C_{i}=Blog_{2}(1+\dfrac{P_{t,i}/L_i}{N})
	\end{split}
	\end{equation}
	where $B$ is the transmission bandwidth of the UAV, $P_{t,i}$ is the UAV transmit power to indoor user $i$, $L_i$ is the path loss between the UAV and indoor user $i$ and $N$ is the noise power. In this paper, we do not explicitly model interference, and instead, implicitly model it as noise.
	
	Let us assume that each indoor user has a channel with bandwidth equals $B/M$, where $M$ is the number of users inside the building and the rate requirement for each user is $v$. Then the minimum power required to satisfy this rate for each user is given by:
	 \begin{equation}
	 \begin{split}
	 P_{t,i,min}=(2^{\frac{v.M}{B}}-1)\star N\star L_i
	 \end{split}
	 \end{equation}
	 Our goal is to find the efficient location of the UAV such that the total transmit power required to satisfy the rate requirement of each indoor user is minimized. The objective function can be represented as:
	 \begin{equation}
	 \begin{split}
	 P=\sum_{i=1}^{M}(2^{\frac{v.M}{B}}-1)\star N\star L_i,\\
	 \end{split}
	 \end{equation}
	 where $P$ is the UAV total transmit power. Since $(2^{\frac{v.M}{B}}-1)\star N$ is constant, our problem can be formulated as:
	  \begin{equation}
	  \begin{split}
	  \min_{x_{UAV},y_{UAV},z_{UAV}} L_{Total}=\sum_{i=1}^{M}L_i~~~~~~~~~~~~~~~~~~~~~~~~~~~~~~~~~\\
	  subject ~to~~~~~~~~~~~~~~~~~~~~~~~~~~~~~~~~~~~~~~~~~~~~~~~~~~~~~~~~~~~\\
	  x_{min}\leq x_{UAV}\leq x_{max},~~~~~~~~~~~~~~~~~~~~~~\\
	  y_{min}\leq y_{UAV}\leq y_{max},~~~~~~~~~~~~~~~~~~~~~~\\
	  z_{min}\leq z_{UAV}\leq z_{max},~~~~~~~~~~~~~~~~~~~~~~\\
	  L_{Total}\leq L_{max}~~~~~~~~~~~~~~~~~~~~~~~~~~~
	   \end{split}
	   \end{equation}
	   Here, the first three constraints represent the minimum and maximum allowed values for $x_{UAV}$, $y_{UAV}$ and  $z_{UAV}$. In the fourth constraint, $L_{max}$ is the maximum allowable path loss and equals $P_{t,max}$$/$$((2^{\frac{v.M}{B}}-1)\star N)$, where $P_{t,max}$ is the maximum transmit power of UAV.
	
	   Finding the optimal placement of UAV is generally  difficult because the problem is non-convex. Therefore, in the next section, we consider two special cases of practical interest and derive efficient solutions for the formulated problems under these cases.
	   \section{Efficient Placement of UAV}
	   \label{sec:optimal}
	   Due to the intractability of the problem, we study the efficient placement of the UAV under two cases. In the first case, we find the minimum transmit power required to cover the building based on the location that has the maximum path loss inside the building. In the second case, we assume that the locations of indoor users are symmetric across the dimensions of each floor, and propose a gradient descent algorithm for finding the efficient location of the UAV. 
	   \subsection{Case One: The worst location in building}
	   \label{case one}
	   In this case, we find the minimum transmit power required to cover the building based on the location that has the maximum path loss inside the building. The location that has the maximum path loss in the building is the location that has maximum $d_{3D,i}$, maximum $\theta_{i}$, and maximum $d_{2D,i}$. The locations that have the maximum path loss are located at the corners of the highest and lowest floors at points $(x_b,0,0)$, $(x_b,y_b,0)$, $(x_b,0,z_b)$ and $(x_b,y_b,z_b)$ (see Figure 1). Since the locations that have the maximum path loss inside the building are the corners of the highest and lowest floors, we place the UAV at the middle of the building ($y_{UAV}$= 0.5$y_b$ and $z_{UAV}$=0.5$z_b$). Then, given the Outdoor-to-Indoor path loss model, we need to find the optimal horizontal point $x_{UAV}$ for the UAV such that the total transmit power required to cover the building is minimized.
	   
	   Now, when the horizontal distance between the UAV and this location increases, the free space path loss also increases as $d_{3D,i}$ increases, while the building penetration loss decreases because we decrease the incident angle $\theta_{i}$. Similarly, when the horizontal distance decreases, the free space path loss decreases as $d_{3D,i}$ decreases, while the building penetration loss increases as the incident angle increases. In Figure~\ref{fig:figure2}, we demonstrate the minimum transmit power required to cover a building of different heights, where the minimum transmit power required to cover the building is given by:
	   \begin{equation}
	   \begin{split}
	   P_{t,min}(dB)=P_{r,th}+L_i
	   \end{split}
	   \end{equation}
	   \begin{equation}
	   \begin{split}
	   P_{r,th}(dB)=N+\gamma_{th}
	   \end{split}
	   \end{equation}
	   Here, $P_{r,th}$ is the minimum received power, $N$ is the noise power (equals -120dBm), $\gamma_{th}$ is the threshold SNR (equals 10dB), $y_b$=50 meters , and $x_b$=20 meters. The numerical results show that there is an optimal horizontal point that minimizes the total transmit power required to cover a building. Also, we can notice that when the height of the building increases, the optimal horizontal distance also increases. This is to compensate the increased building penetration loss due to an increased incident angle.
	   
	    In Theorem 1, we characterize the optimal incident angle $\theta$ that minimizes the transmit power required to cover the building. This helps us finding the optimal horizontal distance between the UAV and the building. 
	   \begin{theorem}
	   	\label{theorem_one}
	   	When we place the UAV at the middle of building , the optimal incident angle $\theta$ that minimizes the transmit power required to cover the building will be equal to $48.654^{o}$ and the optimal horizontal distance between the UAV and the building will be equal to $((\dfrac{0.5z_{b}}{tan(48.654^{o})})^2-(0.5y_{b})^2)^{0.5}-x_{b}$.
	   \end{theorem}
	   \begin{proof}
	     In order to find the optimal horizontal point, we rewrite the equation that represents the path loss in terms of the incident angle ($\theta_{i}$) and the altitude difference between the UAV and the user $i$ ($\Delta h_{i}$):
	   \begin{equation}
	   \begin{split}
	   L_{i}(\Delta h_{i},\theta_{i})=w\log_{10}\dfrac{\Delta h_{i}}{\sin\theta_{i}}+w\log_{10}f_{Ghz}+g_{1}\\
	   +g_{2}+g_{3}(1-\cos\theta_{i})^{2}+g_{4}d_{2D,i}
	   \end{split}
	   \end{equation}
	   We know that the altitude difference between the UAV and the location that has the maximum path loss is constant for a given building. Now, when we take the first derivative with respect to $\theta$ and assign it to zero, we get:
	   \begin{equation}
	   \begin{split}
	   \dfrac{dL(\theta)}{d\theta}=\dfrac{w}{ln10} 
	   \dfrac{\dfrac{-\Delta h.\cos\theta}{\sin^{2}\theta}}{\dfrac{\Delta h}{\sin\theta}}+2g_{3}\sin\theta(1-\cos\theta)=0 ~~~~~~~~\\
	   \dfrac{dL(\theta)}{d\theta}=\dfrac{-w}{ln10}\dfrac{\cos\theta}{\sin\theta}+2g_{3}\sin\theta(1-\cos\theta)=0~~~~~~~~~~~~~~~~ \\
	   \dfrac{w}{ln10}\cos\theta=2g_{3}sin^{2}\theta(1-\cos\theta)~~~~~~~~~~~~~~~~~~~~~~~~~~~~~~~\\
	   \dfrac{w}{ln10}\cos\theta=2g_{3}(1-\cos^{2}\theta)(1-\cos\theta)~~~~~~~~~~~~~~~~~~~~~~~~\\
	   2g_{3}\cos^{3}\theta-2g_{3}\cos^{2}\theta-(\dfrac{w}{ln10}+2g_{3})\cos\theta+2g_{3}=0~~~~~~
	   \end{split}
	   \end{equation}
	   To prove that the function is convex, we take the second derivative and we get:
	   \begin{equation}
	   \begin{split}
	   \dfrac{d^{2}L}{d\theta^{2}}=\dfrac{w}{ln10}\dfrac{1}{\sin^{2}\theta}+2g_{3}\cos\theta(1-\cos\theta)+2g_{3}\sin^{2}\theta>0\\
	   for~0<\theta\leq 90
	   \end{split}
	   \end{equation}
	    Equation (9) has only one valid solution which is $\cos\theta$$=$0.6606, where the non valid solutions are $\cos\theta$$=$1.4117 and $\cos\theta$$=$\textendash1.0723. Therefore, the optimal incident angle between the UAV and the location that has the maximum path loss inside the building will be $48.654^{o}$.
	    
	    In order to find the optimal horizontal distance between the UAV and the building, we apply the pythagorean's theorem. The optimal horizontal distance between the UAV and the location that has maximum path loss inside the building can be represented as: 
	    \begin{equation}
	    \begin{split}
	    d_{H}=((\dfrac{0.5z_{b}}{tan(48.654^{o})})^2-(0.5y_{b})^2)^{0.5}
	    \end{split}
	    \end{equation}
	    In order to find the optimal horizontal distance between the UAV and the building, we subtract $x_{b}$ from $d_{H}$ and we get: 
	     \begin{equation}
	     \begin{split}
	     d_{opt}=((\dfrac{0.5z_{b}}{tan(48.654^{o})})^2-(0.5y_{b})^2)^{0.5}-x_{b}
	     \end{split}
	     \end{equation}
	\end{proof}
	   In Figure~\ref{fig:figure3}, we demonstrate the transmit power required to cover the building as a
	   function of incident angle, we can notice that the optimal angle that we characterize in Theorem 1 gives us the minimum transmit power required to cover the building.	
	   \subsection{Case Two: The locations of indoor users are symmetric across the $xy$ and $xz$ planes}
	   In this case, we assume that the locations of indoor users are symmetric across the $xy$-plane  ((0,0,0.5$z_b$),($x_b$,0,0.5$z_b$) ,($x_b$,$y_b$,0.5$z_b$),(0,$y_b$,0.5$z_b$))) and the $xz$-plane ((0,0.5$y_b$,0), ($x_b$,0.5$y_b$,0), ($x_b$,0.5$y_b$,$z_b$),(0,0.5$y_b$,$z_b$)). First, we prove that $z_{UAV}$=$0.5z_{b}$ and $y_{UAV}$=$0.5y_{b}$ when the locations of indoor users are symmetric across the $xy$ and $xz$ planes, then we will use the gradient descent algorithm to find the efficient $x_{UAV}$ that minimizes the transmit power required to cover the building. Our simulation results show that there is only one local minimum point and the gradient descent algorithm will successfully converge to this point.  
	   \begin{theorem}
	   	When the locations of indoor users are symmetric across the $xy$ and $xz$ planes, the efficient $z_{UAV}$ that minimizes the power required to cover the indoor users will be equal $0.5z_{b}$.
	   \end{theorem}
	   \begin{proof}
	   	Consider that $m_{1}$ represents the users that have altitude lower than the UAV altitude and $m_{2}$ represents the users that have altitude higher than the UAV altitude, then:
	   \begin{equation*}
	   \begin{split}
	   d_{3D,i}=((x_{UAV}-x_i)^2+(y_{UAV}-y_i)^2+(z_{UAV}-z_i)^2)^{0.5} \\
	   ~~~~~~~~~~~~~~~~~~~~~~~~~~~~~~~~\forall z_{UAV}>z_i\\  
	   d_{3D,i}=((x_{UAV}-x_i)^2+(y_{UAV}-y_i)^2+(z_i-z_{UAV})^2)^{0.5} \\
	   ~~~~~~~~~~~~~~~~~~~~~~~~~~~~~~~~\forall z_{UAV}<z_i\\
	   \end{split}
	   \end{equation*}
	   Also, 
	    \begin{equation*}
	    \begin{split}
	   cos_{\theta_i}=\dfrac{((x_{UAV}-x_i)^2+(y_{UAV}-y_i)^2)^{0.5}}{((x_{UAV}-x_i)^2+(y_{UAV}-y_i)^2+(z_{UAV}-z_i)^2)^{0.5}}\\
	   ~~~~~~~~~~~~~~~~~~~~~~~~~~~~~~~~\forall z_{UAV}>z_i\\ 
	   cos_{\theta_i}=\dfrac{((x_{UAV}-x_i)^2+(y_{UAV}-y_i)^2)^{0.5}}{((x_{UAV}-x_i)^2+(y_{UAV}-y_i)^2+(z_i-z_{UAV})^2)^{0.5}}\\
	   ~~~~~~~~~~~~~~~~~~~~~~~~~~~~~~~~\forall z_{UAV}<z_i\\
	   \end{split}
	   \end{equation*}
	   Rewrite the total path loss:
	   \begin{equation*}
	   \begin{split}
	   L_{Total}=~~~~~~~~~~~~~~~~~~~~~~~~~~~~~~~~~~~~~~~~~~~~~~~~~~~~~~~~~~~~~~~~~~~~~~~~~~~~~~\\
	   \sum_{i=1}^{m_1}(wlog_{10}(d_{3D,i})+g_{3}(1-\cos\theta_{i})^{2})+~~~~~~~~~~~~~~~~~~~~~~~~~~~~~~~~~~~~~~~~~~\\
	   \sum_{i=1}^{m_2}(wlog_{10}(d_{3D,i})+g_{3}(1-\cos\theta_{i})^{2})+K~~~~~~~~~~~~~~~~~~~~~~~~~~~~~~~~~~~~~~~\\
	  \end{split}
	  \end{equation*}
	  Where:
	  \begin{equation*}
	  \begin{split}
	   K=\sum_{i=1}^{M}(wlog_{10}f_{Ghz}+g_{1}+g_{2}+g_{4}d_{2D,i})~~~~~~~~~~~~~~~~~~~\\
	   \end{split}
	   \end{equation*}
	   Now, take the derivative with respect to $z_{UAV}$, we get:
	    \begin{equation*}
	    \begin{split}
	    \dfrac{dL_{Total}}{dz_{UAV}}=~~~~~~~~~~~~~~~~~~~~~~~~~~~~~~~~~~~~~~~~~~~~~~~~~~~~~~~~~~~~~~\\
	    \sum_{i=1}^{m_1}\dfrac{w}{ln10}\frac{(z_{UAV}-z_i)}{((x_{UAV}-x_i)^2+(y_{UAV}-y_i)^2+(z_{UAV}-z_i)^2)}~~~\\
	    +2g_{3}.~~~~~~~~~~~~~~~~~~~~~~~~~~~~~~~~~~~~~~~~~~~~~~~~~~~~~~~~~~~~~~~~~~~\\
	    (1-\dfrac{((x_{UAV}-x_i)^2+(y_{UAV}-y_i)^2)^{0.5}}{((x_{UAV}-x_i)^2+(y_{UAV}-y_i)^2+(z_{UAV}-z_i)^2)^{0.5}}).~~~~~\\
	    (\dfrac{((x_{UAV}-x_i)^2+(y_{UAV}-y_i)^2)^{0.5}(z_{UAV}-z_i)}{((x_{UAV}-x_i)^2+(y_{UAV}-y_i)^2+(z_{UAV}-z_i)^2)^{\frac{3}{2}}})+~~~~\\ 
	    \sum_{i=1}^{m_2}\dfrac{w}{ln10}\frac{-(z_{i}-z_{UAV})}{((x_{UAV}-x_i)^2+(y_{UAV}-y_i)^2+(z_i-z_{UAV})^2)}~~~\\
	    +2g_{3}.~~~~~~~~~~~~~~~~~~~~~~~~~~~~~~~~~~~~~~~~~~~~~~~~~~~~~~~~~~~~~~~~~~~\\
	    (1-\dfrac{((x_{UAV}-x_i)^2+(y_{UAV}-y_i)^2)^{0.5}}{((x_{UAV}-x_i)^2+(y_{UAV}-y_i)^2+(z_i-z_{UAV})^2)^{0.5}}).~~~~~\\
	    (\dfrac{-((x_{UAV}-x_i)^2+(y_{UAV}-y_i)^2)^{0.5}(z_i-z_{UAV})}{((x_{UAV}-x_i)^2+(y_{UAV}-y_i)^2+(z_i-z_{UAV})^2)^{\frac{3}{2}}})~~~~
	    \end{split}
	    \end{equation*}
	    
	    Rewrite the $\dfrac{dL_{Total}}{dz_{UAV}}$ again, we have:
	    
	    \begin{equation*}
	    \begin{split}
	  \dfrac{dL_{Total}}{dz_{UAV}} =\sum_{i=1}^{m_1}\dfrac{w}{ln10}\frac{(z_{UAV}-z_i)}{d_{3D,i}^{2}}+~~~~~~~~~~~~~~~~~~~~~\\
	   2g_{3}.(1-\dfrac{((x_{UAV}-x_i)^2+(y_{UAV}-y_i)^2)^{0.5}}{d_{3D,i}}).~~~~~~~~~~~\\
	   (\dfrac{((x_{UAV}-x_i)^2+(y_{UAV}-y_i)^2)^{0.5}(z_{UAV}-z_i)}{d_{3D,i}^{3}})+~~~\\ 
	   \sum_{i=1}^{m_2}\dfrac{w}{ln10}\frac{-(z_i-z_{UAV})}{d_{3D,i}^{2}}+~~~~~~~~~~~~~~~~~~~~~~~~~~~~~~~~~\\
	   2g_{3}.(1-\dfrac{((x_{UAV}-x_i)^2+(y_{UAV}-y_i)^2)^{0.5}}{d_{3D,i}}).~~~~~~~~~~\\
	   (\dfrac{-((x_{UAV}-x_i)^2+(y_{UAV}-y_i)^2)^{0.5}(z{i}-z_{UAV})}{d_{3D,i}^{3}})~~~\\ 
	\end{split}
\end{equation*}
The equation above equals zero when the UAV altitude equals the half of the building height, where the locations of indoor users are symmetric across the $xy$ and $xz$ planes. 
\end{proof}
\begin{theorem}
	When the locations of indoor users are symmetric across the $xy$ and $xz$ planes, the efficient $y_{UAV}$ that minimizes the power required to cover the indoor users will be equal $0.5y_{b}$.
\end{theorem}
The proof of Theorem 3 is similar to that of Theorem 2.
 \begin{figure*}[!t]
 	\centering
 	\includegraphics[scale=0.161]{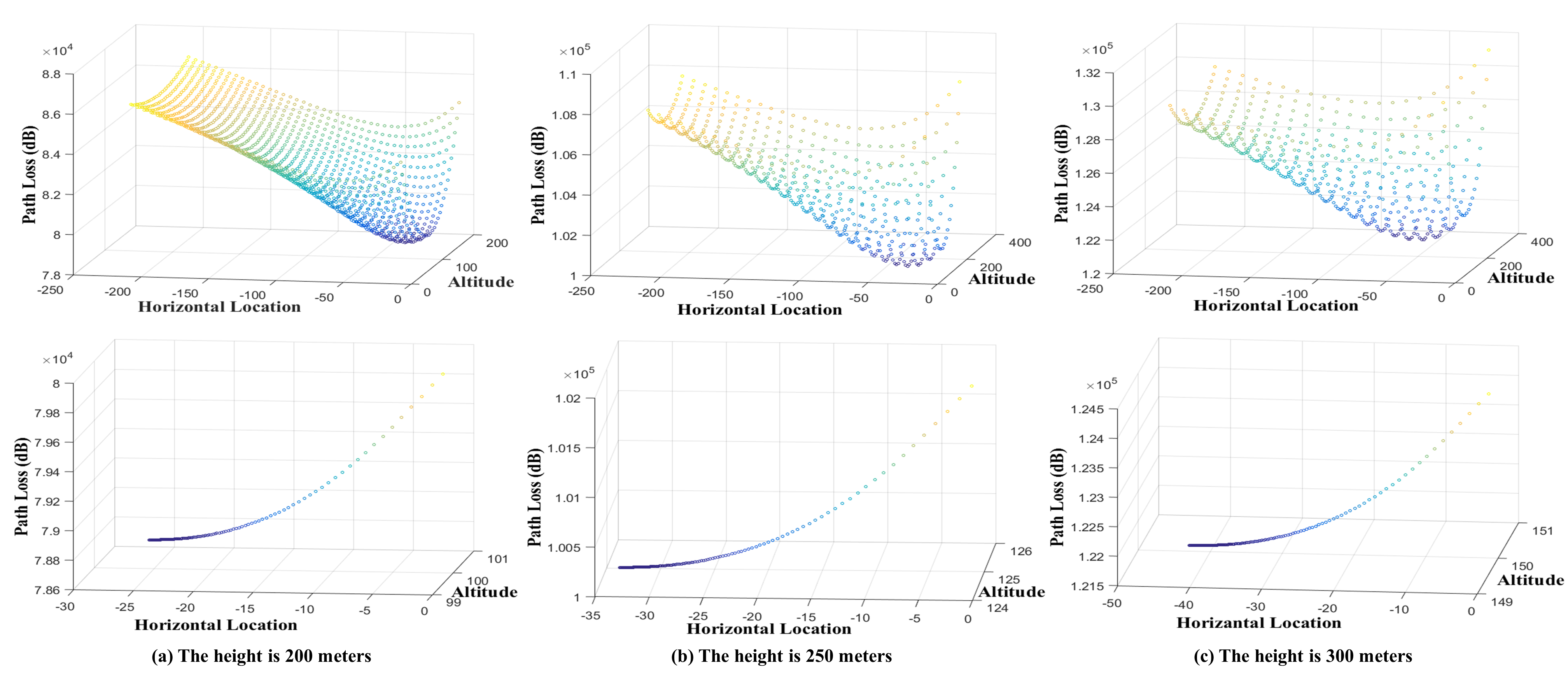}
 	\caption{Total path loss for different building heights}
 	\label{fig:figurefour}
 \end{figure*}
 The question now is how to find the efficient horizontal point $x_{UAV}$ that minimizes the total transmit power.
	   In order to find this point, we use the gradient descent algorithm~\cite{sutton1998reinforcement}:
	   \begin{equation}
	   \begin{split}
	   x_{UAV,n+1}=x_{UAV,n}-a\dfrac{dL_{Total}}{dx_{UAV,n}}~~~~~~~~~~~~~~~~~\\
	   \end{split}
	   \end{equation}
	   Where:\\
	   \begin{equation*}
	   \begin{split}
	     \dfrac{dL_{Total}}{dx_{UAV}}=\sum_{i=1}^{M}\dfrac{w}{ln10}\dfrac{-(x_i-x_{UAV})}{d_{3D,i}^2}+~~~~~~~~~~~~~~~~~~~~~~~~~~~\\
	   2g_{3}.(1-\dfrac{((x_i-x_{UAV})^2+(y_i-y_{UAV})^2)^{0.5}}{d_{3D,i}}).~~~~~~~~~~~~~~~~~\\
	   (\dfrac{(x_i-x_{UAV})d_{3D,i}((x_i-x_{UAV})^2+(y_i-y_{UAV})^2)^{-0.5}}{{d_{3D,i}^2}}-~~\\
	   \dfrac{((x_i-x_{UAV})^2+(y_i-y_{UAV})^2)^{0.5}(x_i-x_{UAV})d_{3D,i}^{-1}}{d_{3D,i}^2})~~~~~~~
	   \end{split}
	   \end{equation*}
	    $a$: the step size.\\
	    $d_{3D,i}$=$((x_i-x_{UAV})^2+(y_i-y_{UAV})^2+(z_i-z_{UAV})^2)^{0.5}$\\
	    
	    The pseudo code of this algorithm is shown in Algorithm 1.
	      
	    \begin{algorithm}
	    	\begin{algorithmic}
	    		\STATE \textbf{Input:}
	    		\STATE The 3D locations of the users inside the building.
	    		\STATE The step size $a$, the step tolerance $\epsilon$.
	    		\STATE The dimensions of the building  $[0,x_b]$ $\times$ $[0,y_b]$ $\times$ $[0,z_b]$.
	    			\STATE The maximum number of iterations $N_{max}$.
	    			\STATE \textbf{Initialize} $x_{UAV}$
	    			\STATE \textbf{For} $n$=1,2,..., $N_{max}$
	    			\STATE  ~~~~~$x_{UAV,n+1}$ $\leftarrow$ $x_{UAV,n}$$-$ $a\dfrac{dL_{Total}}{dx_{UAV,n}}$
	    			\STATE ~~~~~~~~~~~~~\textbf{If} $\lVert$ $x_{UAV,n}$ $-$ $x_{UAV,n+1}$ $\rVert$ $<$ $\epsilon$
	    		\STATE ~~~~~\textbf{Return:} \textbf{$x_{UAV,opt}$} $=$ $x_{UAV,n+1}$
	    		\STATE \textbf{End for}
	    	\end{algorithmic}
	    	\caption{Efficient $x_{UAV}$ using gradient descent algorithm}
	    \end{algorithm}
	    \begin{figure*}[!t]
	    	\centering
	    	\includegraphics[scale=0.16]{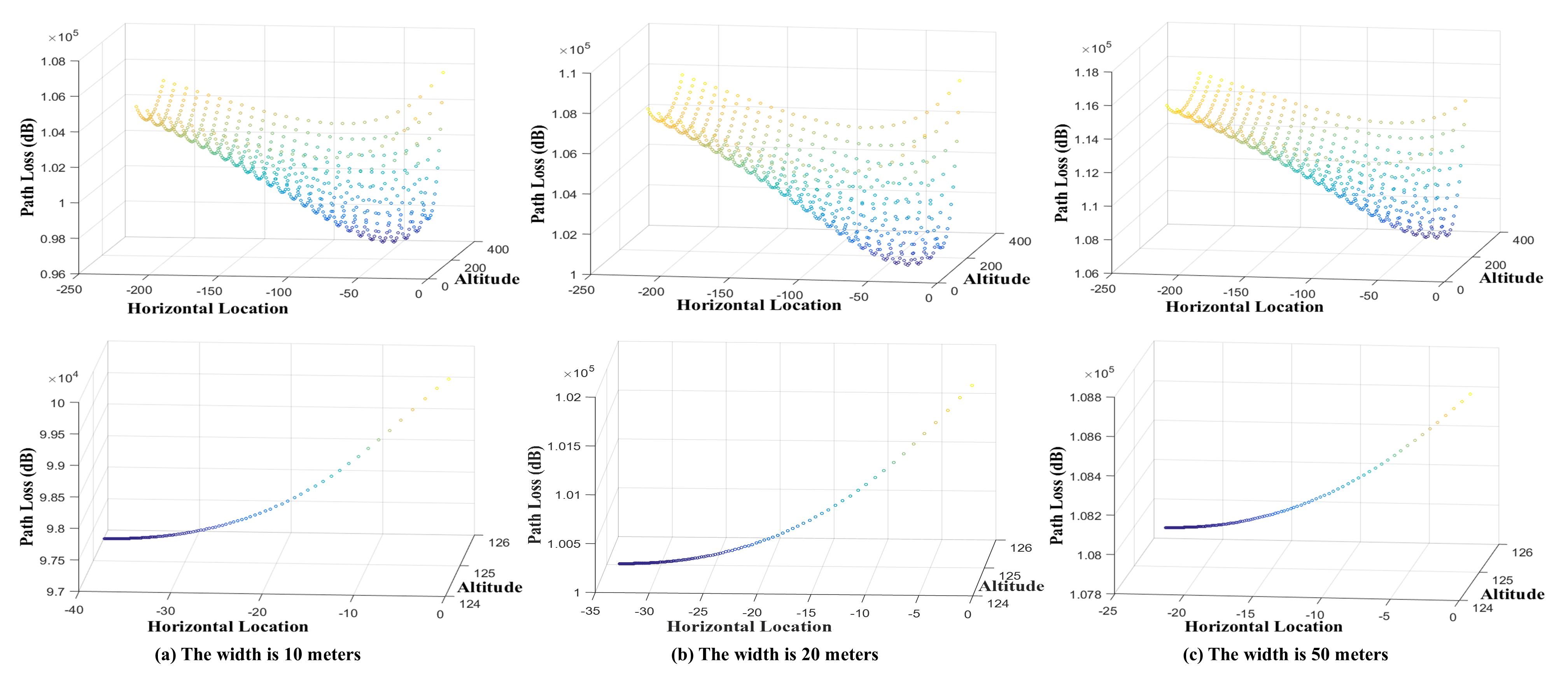}
	    	\caption{Total path loss for different building widths}
	    	\label{fig:figurefive}
	    \end{figure*}
	    
	    \section{Numerical Results}
	    \label{sec:results}
	    In this section, we verify our results for the second case. We assume that each floor contains 20 users. Then, we apply the gradient descent algorithm to find the efficient horizontal point $x_{UAV}$ that minimizes the transmit power required to cover the indoor users. Table I lists the parameters used in the numerical analysis.
	    \begin{table}[!h]
	    	\scriptsize
	    	\renewcommand{\arraystretch}{1.3}
	    	\caption{Parameters in numerical analysis}
	    	\label{table}
	    	\centering
	    	\begin{tabular}{|c|c|}
	    		\hline
	    		Height of building $z_{b}$& 200 meters, 250 meters and 300 meters\\
	    		\hline
	    		Horizontal width of building $x_{b}$ & 10 meters, 20 meters and 50 meters \\
	    			\hline
	    			Vertical width of building $y_{b}$ & 50 meters\\
	    			\hline
	    			Hight of each floor & 5 meters\\
	    			\hline
	    			Step size $a$ & 0.01\\
	    			\hline
	    			Maximum number of iterations $N_{max}$ & 500\\
	    			\hline 
	    			The carrier frequency $f_{Ghz}$ & 2Ghz\\
	    			\hline 
	    			Number of users in each floor & 20 users\\
	    			\hline
	    		\end{tabular}
	    	\end{table}
	    
	    In Figure~\ref{fig:figurefour}, we find the efficient horizontal points for a building of different heights. In the upper part of the figures, we find the total path loss at different locations ($x_{UAV}$,0.5$y_{b}$,$z_{UAV}$) and in the lower part of the figures, we find the efficient horizontal point $x_{UAV}$ that results in the minimum total path loss using the gradient descent algorithm. As can be seen from the figures, when the height of the building increases, the efficient horizontal point $x_{UAV}$ increases. This is to compensate the increased building penetration loss due to an increased incident angle. 
	    
	    In Figure~\ref{fig:figurefive}, we investigate the impact of different building widths (i.e., $x_b$). We fix the building height to be 250 meters and vary the building width. As can be seen from the figures, when the building width increases, the efficient horizontal distance decreases. This is to compensate the increased indoor path loss due to an increased building width.
	    
	    In~\cite{hazim2017}, we validate the simulation results by using the particle swarm optimization algorithm and study the problem when the locations of indoor users are uniformly distributed in each floor.
	    \section{Conclusion}
	    \label{sec:Conclusion}
	    In this paper, we study the problem of providing wireless coverage for users inside a high-rise building using a single UAV. First, we demonstrate why the Air-to-Ground path loss model is not appropriate for considering indoor users with 3D locations. Then, we present the Outdoor-to-Indoor path loss model, show the tradeoff in this model, and study the problem of minimizing the transmit power required to cover the building. Due to the intractability of the problem, we study the efficient placement of the UAV under two cases. In the first case, we find the minimum transmit power required to cover the building based on the location that has the maximum path loss inside the building. In the second case, we assume that the locations of indoor users are symmetric across the $xy$ and $xz$ planes and we use the gradient descent algorithm to find the efficient placement of the UAV. In order to model more realistic scenarios, we will consider different types of user distribution in our future work. We will also study the problem of providing wireless coverage using multiple UAVs. 
	    \section*{Acknowledgment}
	    This work was supported in part by the NSF under Grants CNS-1647170 and CNS-1651947.

	\bibliographystyle{IEEEtran}
	\bibliography{UAVpath}

\begin{thebibliography}{10}
\providecommand{\url}[1]{#1}
\csname url@samestyle\endcsname
\providecommand{\newblock}{\relax}
\providecommand{\bibinfo}[2]{#2}
\providecommand{\BIBentrySTDinterwordspacing}{\spaceskip=0pt\relax}
\providecommand{\BIBentryALTinterwordstretchfactor}{4}
\providecommand{\BIBentryALTinterwordspacing}{\spaceskip=\fontdimen2\font plus
\BIBentryALTinterwordstretchfactor\fontdimen3\font minus
  \fontdimen4\font\relax}
\providecommand{\BIBforeignlanguage}[2]{{%
\expandafter\ifx\csname l@#1\endcsname\relax
\typeout{** WARNING: IEEEtran.bst: No hyphenation pattern has been}%
\typeout{** loaded for the language `#1'. Using the pattern for}%
\typeout{** the default language instead.}%
\else
\language=\csname l@#1\endcsname
\fi
#2}}
\providecommand{\BIBdecl}{\relax}
\BIBdecl

\bibitem{bupe2015relief}
P.~Bupe, R.~Haddad, and F.~Rios-Gutierrez, ``Relief and emergency communication
  network based on an autonomous decentralized uav clustering network,'' in
  \emph{SoutheastCon 2015}.\hskip 1em plus 0.5em minus 0.4em\relax IEEE, 2015,
  pp. 1--8.

\bibitem{bor2016efficient}
R.~I. Bor-Yaliniz, A.~El-Keyi, and H.~Yanikomeroglu, ``Efficient 3-d placement
  of an aerial base station in next generation cellular networks,'' in
  \emph{Communications (ICC), 2016 IEEE International Conference on}.\hskip 1em
  plus 0.5em minus 0.4em\relax IEEE, 2016, pp. 1--5.

\bibitem{al2014modeling}
A.~Al-Hourani, S.~Kandeepan, and A.~Jamalipour, ``Modeling air-to-ground path
  loss for low altitude platforms in urban environments,'' in \emph{2014 IEEE
  Global Communications Conference}.\hskip 1em plus 0.5em minus 0.4em\relax
  IEEE, 2014, pp. 2898--2904.

\bibitem{mozaffari2015drone}
M.~Mozaffari, W.~Saad, M.~Bennis, and M.~Debbah, ``Drone small cells in the
  clouds: Design, deployment and performance analysis,'' in \emph{IEEE Global
  Communications Conference (GLOBECOM)}, 2015, pp. 1--6.

\bibitem{mozaffari2016optimal}
M.~Mozaffari\vspace{0mm}, W.~Saad, M.~Bennis, and M.~Debbah, ``Optimal
  transport theory for power-efficient deployment of unmanned aerial
  vehicles,'' \emph{IEEE International Conference on Communications (ICC),
  Kuala Lumpur, Malaysia,}, 2016.

\bibitem{mozaffari2016efficient}
M.~Mozaffari, W.~Saad, M.~Bennis, and M.~Debbah, ``Efficient deployment of
  multiple unmanned aerial vehicles for optimal wireless coverage,''
  \emph{arXiv preprint arXiv:1606.01962}, 2016.

\bibitem{kalantari2016number}
E.~Kalantari, H.~Yanikomeroglu, and A.~Yongacoglu, ``On the number and 3d
  placement of drone base stations in wireless cellular networks,'' in
  \emph{IEEE Vehicular Technology Conference}, 2016, pp. 18--21.

\bibitem{shakhatreh2016continuous}
H.~Shakhatreh, A.~Khreishah, J.~Chakareski, H.~B. Salameh, and I.~Khalil, ``On
  the continuous coverage problem for a swarm of uavs,'' in \emph{Sarnoff
  Symposium, 2016 IEEE 37th}.\hskip 1em plus 0.5em minus 0.4em\relax IEEE,
  2016, pp. 130--135.

\bibitem{ericsson}
``Ericsson report optimizing the indoor experience,
  http://www.ericsson.com/res/docs/2013/real-performance-indoors.pdf,'' 2013.

\bibitem{alcatel}
``In-building wireless: One size does not fit all,
  http://www.alcatel-lucent.com/solutions/in-building/in-building-infographic.''

\bibitem{cisco}
``Cisco service provider wi-fi: A platform for business innovation and revenue
  generation,
  http://www.cisco.com/c/en/us/solutions/collateral/service-provider/service-provider-wi-fi/solution{\_}overview{\_}c22{\_}642482.html.''

\bibitem{comm}
``Using high-power das in high-rise buildings,
  http://www.commscope.com/docs/using-high-power-das-in-high-rise-buildings-an-318376-ae.pdf.''

\bibitem{amplitic}
``Coverage solution for high-rise building,
  http://www.amplitec.net/products-2-coverage-solution-for-high-rise-building.html.''

\bibitem{zhang2016study}
S.~Zhang, Z.~Zhao, H.~Guan, and H.~Yang, ``Study on mobile data offloading in
  high rise building scenario,'' in \emph{Vehicular Technology Conference (VTC
  Spring), 2016 IEEE 83rd}.\hskip 1em plus 0.5em minus 0.4em\relax IEEE, 2016,
  pp. 1--5.

\bibitem{series2009guidelines}
M.~Series, ``Guidelines for evaluation of radio interface technologies for
  imt-advanced,'' \emph{Report ITU}, no. 2135-1, 2009.

\bibitem{sutton1998reinforcement}
R.~S. Sutton and A.~G. Barto, \emph{Reinforcement learning: An
  introduction}.\hskip 1em plus 0.5em minus 0.4em\relax MIT press Cambridge,
  1998, vol.~1, no.~1.

\bibitem{hazim2017}
H.~Shakhatreh, A.~Khreishah, A.~Alsarhan, I.~Khalil, A.~Sawalmeh, and
  O.~Noor~Shamsiah, ``Efficient 3d placement of a uav using particle swarm
  optimization,'' in \emph{The International Conference on Information and
  Communication Systems (ICICS 2017) (accepted)}.

\end{thebibliography}

\end{document}